\documentclass[review]{elsarticle}
\usepackage{lineno,hyperref}
\modulolinenumbers[5]

\usepackage{amsmath,latexsym}
\usepackage{amsfonts}
\usepackage{amssymb}
\usepackage{amssymb}
\usepackage{graphicx}
\usepackage{epstopdf}
\usepackage{subfigure}
\usepackage{setspace}
\usepackage{color}
\usepackage{setspace}
\usepackage{algpseudocode}
\usepackage{algorithm}
\usepackage{ulem, enumerate}
\usepackage[colorinlistoftodos]{todonotes}

\usepackage{tikz}
\newtheorem{thrm}{Theorem}
\newtheorem{claim}{Claim}
\newtheorem{lemm}{Lemma}

\newtheorem{defn}{Definition}

\newtheorem{corollary}{Corollary}

\def\qedbox#1#2{\vbox{\hrule height.2pt
  \hbox{\vrule width.2pt height#2pt \kern#1pt \vrule width.2pt}
  \hrule height.2pt}}
\def\qed{\hfill \quad\qedbox46\newline\smallbreak}

\def\s#1{\mbox{\boldmath $#1$}}
\def\floor#1{\lfloor #1 \rfloor}

\def\+{\!+\!}
\def\-{\!-\!}

\def\m{\!-\!}

\def\itbf#1{\textit{\textbf{#1}}}

\def\bproc{{\bf procedure\ }}
\def\bfunc{{\bf function\ }}

\def\bfor{{\bf for\ }}

\def\bto{{\bf to\ }}
\def\bdownto{{\bf downto\ }}

\def\band{{\bf and\ }}

\def\bdo{{\bf do\ }}

\def\bif{{\bf if\ }}
\def\band{{\bf and\ }}
\def\bthen{{\bf then\ }}
\def\belse{{\bf else\ }}
\def\belsif{{\bf else if\ }}

\def\breturn{{\bf return\ }}

\def\bstep{{\bf step\ }}

\def\la{\leftarrow}

\def\qq{\qquad}
\def\com#1{{\bf $\triangleright$}\hspace{6pt}{\sl #1}}

\def\pref(#1,#2){$#1$ is a prefix of $#2$}
\def\suff(#1,#2){$#1$ is a suffix of $#2$}

\def\reg(#1,#2){$#2$ is $#1$-regular}
\def\notreg(#1,#2){$#2$ is not $#1$-regular}

\def\true{\tt{true}}
\def\false{\tt{false}}
\def\UPDATE\_F{\tt{UPDATE\_F}}

\def\MP{\mbox{MP}}
\def\FP{\mbox{FP}}
\def\FS{\mbox{FS}}

\def\O{\mathcal{O}}
\def\Q'{\tt{Q'}}

\newif\ifShow
\Showfalse

\newtheorem{conj}{Conjecture}

\algnewcommand{\LineComment}[1]{\State \(\triangleright\) \normalfont{\sl #1}}
\algtext*{EndWhile}
\algtext*{EndIf}
\algtext*{EndFor}
\algtext*{EndFor}
\algtext*{EndProcedure}

\newcommand{\xRpos}{P_{\s{x_R}}}

\begin{document}

\begin{frontmatter}

\title{A New Approach to Regular \& Indeterminate Strings}

\author[ufu]{Felipe A. Louza}
\ead{louza@ufu.br}
\address[ufu]{Faculty of Electrical Engineering, Federal University of Uberl\^andia, Brazil}
\author[macmaster]{Neerja Mhaskar\corref{mycorrespondingauthors}}
\cortext[mycorrespondingauthors]{Corresponding author}
\ead{pophlin@mcmaster.ca}
\address[macmaster]{Department of Computing and Software, McMaster University, Canada}

\author[macmaster,kcl,murdoch]{W.~F.~Smyth}
\ead{smyth@mcmaster.ca}
\address[kcl]{Department of Informatics, King's College London, UK}
\address[murdoch]{School of Engineering \& Information Technology, Murdoch University, Australia}

\begin{abstract}
In this paper we propose a new, more appropriate
definition of \itbf{regular} and \itbf{indeterminate} strings. A \itbf{regular} string is 
one that is ``isomorphic'' to a string whose entries
all consist of a single letter, but which nevertheless may itself
include entries containing multiple letters.
A string that is not regular is said to be \itbf{indeterminate}.
We begin by proposing a new model for the representation of strings, regular or indeterminate, then go on to
describe a linear time algorithm to determine whether or not a string $\s{x} = \s{x}[1..n]$ is regular and, if so,
to replace it by a lexicographically least (lex-least) string \s{y} whose entries are all single letters.  Furthermore, we connect the regularity of a string to the transitive closure problem on a graph, which in our special case can be efficiently solved.
We then introduce the idea of a feasible palindrome array $\MP$ of a string,
and prove that every feasible $\MP$ corresponds to
some (regular or indeterminate) string.
We describe an algorithm that constructs a string \s{x} corresponding to given feasible $\MP$,
while ensuring that whenever possible \s{x} is regular and if so, then lex-least.
A final section outlines new research directions suggested by
this changed perspective on regular and indeterminate strings.
\end{abstract}

\begin{keyword}
indeterminate string \sep degenerate string \sep generalized string \sep transitive closure \sep palindrome \sep maximal palindrome array \sep Manacher's condition \sep reverse engineering \sep algorithm \sep stringology
\vspace{-3mm}
\end{keyword}

\end{frontmatter}


\section{Introduction}
\label{sect-intro}
Under various names and guises, the idea of a string
as something other than a sequence of single letters has been discussed
for almost half a century.
In 1974 Fischer \& Paterson \cite{FP74} studied pattern-matching
on strings \s{x} whose entries could be \itbf{don't-care} letters; that is, 
letters matching any single letter in the alphabet $\Sigma$
on which the string is defined,
hence matching every position in \s{x}.
(Thus the don't care letter $*$ would match every letter
in a given alphabet $\Sigma = \{a,b,c\}$.)
In 1987 Abrahamson \cite{A87} extended this model by considering
pattern-matching on \itbf{generalized} strings whose entries could be
non-empty subsets of $\Sigma$.
(Thus entries $\{b\}$, $\{a,b\}$, $\{b,c\}$ would match each other 
due to the shared $b$.)

Both of these models have been intensively studied in this century, with interest fueled by applications, especially to computational biology, where as a byproduct of the DNA sequencing process and in other contexts, indeterminate letters such as $\{a,t\}$, $\{c,g\}$ or $\{a,c,g,t\}$ may arise.
The properties of strings with don't care letters,
also called \itbf{holes}, have been intensively studied by Blanchet-Sadri,
culminating in a 2008 book \cite{BS08} on \itbf{strings with holes},
also called \itbf{partial words},
that summarized the contributions of her many previous research articles.
In 2009 a periodicity lemma for partial words was formulated in \cite{SW09a}.

Meanwhile Iliopoulos and various collaborators
(see, for example, \cite{IMR08,IRVV10,IKP17})
studied pattern-matching
and other problems, especially related to bioinformatics applications,
on generalized strings (now renamed \itbf{degenerate strings}, or even elastic degenerate strings).
Another sequence of papers \cite{HS03,HSW05,HSW06,HSW08,SW09}
studied other pattern-matching algorithms on generalized strings
(renamed \itbf{indeterminate strings}, with \itbf{indeterminate letters} such as $\{a,b\}$ or $\{b,c\}$ in some positions).
In addition, other aspects of indeterminate strings have been studied,
such as their border arrays \cite{NRR12}, cover arrays \cite{BRS09},
prefix arrays \cite{ARS15,CRSW15} and suffix arrays \cite{BIST03,NRR12}.
The computation of an indeterminate string from its prefix graph
was recently considered in \cite{HRSS18}.

The traditional approach to defining an indeterminate string only relies on the scope of its letters irrespective of the matching relationship between them. Therefore, a string $\s{x}=a\{a,c\} b \{a, d\} bb$ can be seen as an indeterminate string, when in essence it can be reduced to an ``isomorphic'' string $\s{x'}=aababb$ that is regular (see Section~\ref{sect-indet} for definitions).
It is evident that algorithms designed for indeterminate strings are slower and at times require more space for processing. Therefore, it is important to identify strings that are ``truly'' indeterminate. In this paper we redefine an indeterminate string in a context
that we believe captures the idea in a more appropriate way --- at once more general and more precise.

The outline of the paper is as follows: in Section~\ref{sect-indet} we present the necessary definitions, including in particular a model for
the representation of strings that efficiently incorporates both regular and indeterminate.
In Section \ref{sect-alg}, we apply this model to describe
a linear time algorithm that determines whether or not a given
string is regular and, if so, returns an isomorphic
lexicographically least (lex-least) regular string, whose entries are all single letters.  We go on to present an interesting connection between strings defined using our approach and the transitive closure of a relation, showing that for some special cases of a relation, its transitive closure can be computed in quadratic time.
In Section~\ref{sect-maxpal} we introduce a ``Reverse Engineering'' problem \cite{FLRSSY99,FGLRSSY02,CCR09,DFHSS18}
for both regular and indeterminate strings (in our new formulation).
We first show that every feasible palindrome array MP does indeed
correspond to some string, regular or indeterminate.
Then, following \cite{IIBT10}, we describe an algorithm that,
given a feasible palindrome array MP,
computes a lex-least regular string corresponding to
MP, whenever that is possible; and, if not,
returns a corresponding indeterminate string.
Section~\ref{sect-probs} discusses open problems and questions arising
from this new approach.

\section{Regular \& Indeterminate Strings}
\label{sect-indet}
A \itbf{letter} $\ell$ is a finite list of $s$ distinct \itbf{characters} $c_1,c_2,\ldots,c_s$, each drawn from a set $\Sigma$
of size $\sigma = |\Sigma|$
called the \itbf{alphabet}.
In the case that $\Sigma$ is ordered,
$\ell$ is said to be in \itbf{normal form} if its characters occur in the ascending order determined by $\Sigma$.
The integer $s = s(\ell)$ is called the \itbf{scope} of $\ell$.
If $s = 1$, $\ell$ is said to be \itbf{regular}, otherwise \itbf{indeterminate}.
Two letters $\ell_1,\ell_2$ are said to \itbf{match},
written $\ell_1 \approx \ell_2$, if and only if
$\ell_1 \cap \ell_2 \ne \emptyset$.
In the case that matching $\ell_1$ and $\ell_2$ are both regular,
we may write $\ell_1 = \ell_2$.
  
For $n \ge 1$, a \itbf{string} $\s{x} = \s{x}[1..n]$ is a sequence 
$\s{x}[1],\s{x}[2],\ldots,\s{x}[n]$ of letters,
where $n = |\s{x}|$ is the \itbf{length} of \s{x},
and every $i \in \{1,\ldots,n\}$ is a \itbf{position} in \s{x}.
If every letter in \s{x} is in normal form, then \s{x} itself
is said to be in \itbf{normal form}.
A tuple $T = (i,j_1,j_2)$ of distinct positions $i,j_1,j_2$ in \s{x}
such that $$\s{x}[j_1] \approx \s{x}[i] \approx \s{x}[j_2]$$
is said to be a \itbf{triple}.
A triple $T$ is \itbf{transitive} if $\s{x}[j_1] \approx \s{x}[j_2]$,
otherwise \itbf{intransitive}.
If every triple $T$ in \s{x} is transitive,
then we say that \s{x} is \itbf{regular}; otherwise, \s{x} is \itbf{indeterminate}.
The \itbf{scope} of \s{x} is given by
$$S(\s{x}) = \max_{i\in \{1,\dots ,n\}} s(\s{x}[i]).$$
 
Two strings \s{x} and \s{y} of equal length $n$
are said to be \itbf{isomorphic} if and only if
for every $i,j \in$ $\{1,\dots,n\}$,
\begin{equation}
\label{rule2}
\s{x}[i] \approx \s{x}[j] \Longleftrightarrow \s{y}[i] \approx \s{y}[j].
\end{equation}
Note that isomorphic strings may 
superficially have quite
different representations; for example,
$\s{y} = \{a,b\}\{b,c\}\{a,c\}$ is isomorphic to $\s{x} = aaa$.
Thus we see that the definition given here of regularity of a string \s{x}
is more general, including many more strings,
than the usual definition (scope $= 1$) given in the literature.
However:
\begin{lemm}
\label{lemmreg}
Every regular string is isomorphic to a string of scope 1.
\end{lemm}
\begin{proof} 
In \cite{CRSW15,HRSS18} and elsewhere it is observed that a regular string
$\s{x}[1..n]$ is equivalent to a graph $\mathcal{G}$ on vertices $\{1,\dots,n\}$ which is a
collection of cliques with the property that vertices
$i$ and $j \ne i$ are in the same clique if and only if 
$\s{x}[i] \approx \s{x}[j]$.
Now if, for every clique $\mathcal{C}$ of $\mathcal{G}$,
we choose a unique regular letter $\ell_{\mathcal{C}}$ and,
corresponding to every position $i$ in $\mathcal{C}$,
make the assignment $\s{x}[i] \la \ell_{\mathcal{C}}$,
we thus construct a string on regular letters
equivalent to $\mathcal{G}$; that is,
a string of scope 1.  \qed
\end{proof}
We remark that constructing the graph $\mathcal{G}$ described in Lemma~\ref{lemmreg}, then checking whether or not it is a collection of cliques, is a trivial and obvious approach to determining the regularity of \s{x}.  In this paper we propose more efficient and interesting algorithms.

The following result is an easy corollary of Lemma~\ref{lemmreg}:
\begin{corollary}
\label{lemmmap}
Given a regular string $\s{x}[1..n]$, then, corresponding to every triple $(i, j_1, j_2)$, we can assign a regular letter to $\s{y}[i], \s{y}[j_1], \s{y}[j_2]$ in such a way that the resulting string $\s{y}[1..n]$ is isomorphic to $\s{x}[1..n]$. \qed
\end{corollary}
For example, $\s{y} = ccdcdd$ is isomorphic to 
$\s{x} = a\{a,c\} b \{a, d\} bb$.

\subsection{A Model for the Representation of Indeterminate Strings}
The preceding result suggests a need for an algorithm to determine whether or not an apparently indeterminate string is nevertheless actually regular, even though,
like the string $\s{x} = a\{a,c\} b \{a, d\} bb$ mentioned in the Introduction, it contains indeterminate letters.
We describe such an algorithm in the next section,
but first we must deal with the following issue:
in order to design algorithms on indeterminate strings \s{x}, we need to have some practical representation of such strings --- in particular, of indeterminate letters.
This requires us to deal in some realistic way with the fact that, on an alphabet of size $\sigma$, there are $2^{\sigma}\- 1$ possible nonempty distinct letters.
For $\sigma = 2$, we have only $a, b, \{a,b\}$, but for $\sigma = 26$,
dealing efficiently with $2^{26} - 1$ distinct letter sets
is beyond current processing capabilities.
Fortunately, in practice, such cases do not arise:
there may be ``many'' letter sets but, for large $\sigma$, not $2^\sigma\- 1$ of them.

Thus we need a model of computation that deals efficiently with indeterminate letters even when their number is large --- though not ``too large''.
For given $\sigma$,
we shall accordingly assume an integer alphabet with letters $\ell \in \{1,2,\ldots,\sigma\}$ and we shall interpret letter
$\ell = 0$ as denoting a hole or don't-care.
Otherwise, $\ell > \sigma$ represents an indeterminate letter, that also determines the position $\sigma' = \ell\- \sigma$ in an array $I[1..\sigma^*]$, where $\sigma^*$ is the maximum number of indeterminate letters (in addition to the don't care)
allowed to occur in \s{x}.  Each entry $I[\sigma']$ is a pair $(s(\ell),loc)$ that gives the scope $s(\ell)$ of the indeterminate letter $\ell$ and the starting location $loc$ in an array $L$, consisting of integers in the range $\{1\dots\sigma\}$, of the subarray specifying the $s(\ell)$ regular letters that constitute $\s{x}[i]$:
$$(\ell_1, \ell_2, \ldots, \ell_{s(\ell)}).$$
Here the $\ell_j, 1 \le j \le s(\ell)$, are distinct elements of $\Sigma$ in normal form, thus satisfying
$\ell_1 < \ell_2 < \cdots < \ell_{s(\ell)}$. For instance, consider $\s{x}=aac\{a,c\}gta\{g,t\}\{a,c\}\{g,t\}$ over $\Sigma=\{a,c,g,t\}$. To represent it using the above encoding, we map $a \rightarrow 1, c \rightarrow 2, g \rightarrow 3, t \rightarrow 4$. Since we have two indeterminate letters in $\s{x}$, $\sigma^*=2$, $0\leq \ell \leq 6$, and $\{a,c\} \rightarrow 5$ and $\{g, t\} \rightarrow 6$. Further, $I=[(2,1), (2,3)]$, $L=[1,2,3,4]$, and the encoding of $\s{x}$ is $1125341656$.

To fix the ideas, we suppose here that $\ell$, hence every entry in array $L$, is constrained to one byte (8 bits) of storage.  Thus, for example, to handle English text, we should suppose that $\sigma = 77$ (26 lower case letters, 26 upper case,
10 numerical digits, 14 punctuation symbols, and space), leaving available values $\ell \in 78..255$: $\sigma^* = 178$ distinct indeterminate letters.
(Doubling storage for $\ell$ to two bytes permits 65,458 indeterminate English letters to be defined.)

Another example: for DNA strings on $\Sigma = \{a,c,g,t\}$, $\sigma = 4$, we could define
$a \rightarrow 1, c \rightarrow 2, g \rightarrow 3, 
t \rightarrow 4$, use zero for the don't care and 5--14
for the other 10 possible combinations of letters:
$$ac,ag,at,cg,ct,gt,acg,act,agt,cgt.$$
Then for DNA a half byte (4 bits) for $\ell$ would be sufficient to identify all the indeterminate letters, and the associated array $L$ would require at most 24 half bytes to represent them.

Other models to represent indeterminate strings have been proposed. For example, in~\cite{prochazka19} all the non-empty letters (both regular and indeterminate) over the DNA alphabet $\Sigma_{DNA}=\{A,C,G,T\}$ are mapped to the IUPAC symbols $\Sigma_{IUPAC}=\{A,C,G,T,R,Y,S,W,K,M,B,D,H,V,N\}$, to construct an isomorphic regular string of scope $1$ over $\Sigma_{IUPAC}$. Further, in~\cite{btt2013} each symbol in the DNA alphabet is represented as a 4-bit integer power of $2$ ($2^i$, with $i \in \{0,1,2,3\}$). A non-empty indeterminate letter over $\Sigma_{DNA}$ is represented as $\Sigma_{\{s \in \mathcal{P}(\Sigma)\}}s$. Then instead of using the natural order on integers it uses a Gray code~\cite{gray53} (also known as the reflected binary code) to order indeterminate letters over $\Sigma_{DNA}$. Note that with the Gray code two successive values differ only by one bit, such as 1100 and 1101, which enables minimizing the number of separate intervals associated with each of the four symbols of $\Sigma_{DNA}$.
The models described here are of interest primarily for representing indeterminate strings over the DNA alphabet. However, the same approach is applicable to any alphabet.

\section{Determining the regularity of \s{x}}
\label{sect-alg}
In this section we discuss a function REGULAR
(Figure~\ref{fig-regular})
that determines whether or not a given string \s{x} is regular.
REGULAR deals with strings \s{x}
of scope $S(\s{x}) > 1$,
since otherwise, for $S(\s{x}) = 1$,
the regularity of \s{x} is immediate.
As a byproduct, in the case that \s{x} is regular, 
REGULAR computes the lex-least
string \s{y} of scope 1 on an integer alphabet that is isomorphic to \s{x}.

Before getting into the details of REGULAR, we need to consider the overall strategy.
An issue that arises is this: in order to determine whether or not $\s{x}[1..n]$ is regular, it is necessary to look at all triples to determine whether or not one of them is intransitive, a process that it seems must require $n\choose 2$ letter comparisons, thus $\O(S(\s{x})n^2)$ time.
However, it suffices to apply REGULAR to a \itbf{reduced string} $\s{x_R} = \s{x_R}[1..m]$ consisting of exactly one occurrence of each of the $m$ distinct letters (both regular and indeterminate) in \s{x}.
Thus, in view of our constraint on $\ell$, $m \le 256$.
Clearly \s{x_R} is regular if and only if \s{x} is regular.
Consequently triple-testing time reduces to $\O(m^2)$, where of course $m << n$.  (Indeed, even for the extreme two-byte case mentioned above, with $m = 65,536$, $m^2$ is only 4 billion or so, thus still $\O(n)$ in many cases.)

\begin{figure}[ht!]
{\leftskip=0.5in\obeylines\sfcode`;=3000
\bfunc $REGULAR(\s{x},n;\s{y},\sigma')$ : Boolean
construct $\s{x_R}$ \com{contains each distinct letter of $\s{x}$ exactly once}
$m \la |\s{x_R}|$
$ret \la regular_{min}(\s{x_R}, m, \s{y_R}, \sigma')$
\bif $ret$ \bthen 
\qq construct $\s{y}$ from $\s{y_R}$
$REGULAR \la ret$;\ \breturn{}

}
\caption{This Boolean function determines whether ($REGULAR = \true$) or not ($REGULAR = \false$) a string $\s{x}[1..n]$ of scope $S(\s{x}) > 1$ on alphabet $\Sigma$ is regular.  In the former case, on exit the string \s{y} is the lex-least regular string of scope 1 on an integer alphabet $\Sigma' = \{1,2,\ldots,\sigma'\}$ that is isomorphic to \s{x}. 
}
\label{fig-regular}
\end{figure}

The first step of REGULAR therefore computes \s{x_R}, a process performed in $\Theta(n)$ time by a single scan of \s{x}, using a Boolean array $found[1..256]$ to indicate whether ($found[j] = \true$) or not ($found[j] = \false$) the letter $j = \s{x}[i]$ has occurred previously in $\s{x}[1..i-1]$.
Note that in this process, the letters are placed in \s{x_R} in the order of their first occurrence in \s{x}.
Thus:
\begin{claim}
\label{claim-basic}
In cases of practical interest --- that is, when $n$ greatly exceeds $m$ --- the replacement of $\s{x} = \s{x}[1..n]$ by the reduced string $\s{x_R} = \s{x_R}[1..m]$, where $m$ is the number of distinct letters in \s{x}, permits the regularity of \s{x} to be determined in $\O(n)$ time. \qed
\end{claim}
Indeed, in many cases that arise in practice, the status of \s{x} follows from conditions that are more easily evaluated.  For instance, \s{x} is indeterminate if
\begin{itemize}
    \item it contains all regular letters in $\Sigma$ and at least one indeterminate letter; or
    \item it contains a don't-care and at least two regular letters; or
    \item it contains an indeterminate letter $\ell$ as well as any two characters of $\ell$.
\end{itemize}

Having constructed $\s{x_R}$, REGULAR executes
the Boolean function $regular_{min}$ (Figure~\ref{fig-regularmin}), which returns $\true$ if and only if $\s{x_R}$ is regular.
In that case,
$regular_{min}$ also outputs a string $\s{y_R}$ of scope 1 on an integer alphabet $\Sigma'$ that is isomorphic to \s{x_R}.
As explained below, it is then straightforward to compute a lex-least string \s{y} of scope 1 isomorphic to \s{x}.

\begin{figure}[htb]
{\leftskip=0.5in\obeylines\sfcode`;=3000
\bfunc $regular_{min}(\s{x_R};m,\s{y_R},\sigma')$ : Boolean
$\s{y_R}[1..m] \la 0^m;\ \sigma' \la 0$
\bfor $i = 1$ \bto $m$ \bdo
\qq \bif $\s{y_R}[i] = 0$ \bthen 
\qq\qq $\sigma' \la \sigma'\+ 1;\ \s{y_R}[i] \la \sigma'$
\qq\qq \bfor $j \la i\+ 1$ \bto $m$ \bdo
\qq\qq\qq \bif $\s{x_R}[j] \approx \s{x_R}[i]$ \bthen
\qq\qq\qq\qq \bif $\s{y_R}[j] = 0$ \bthen{} $\s{y_R}[j] \la \sigma'$
\qq\qq\qq\qq \belse $regular_{min} \la \false;\ \breturn$
\qq \belse
\qq\qq \bfor $j \la i\+ 1$ \bto $m$ \bdo
\qq\qq\qq \bif $\s{y_R}[j] = \s{y_R}[i]$ \bthen
\qq\qq\qq\qq \bif $\s{x_R}[j] \not\approx \s{x_R}[i]$ \bthen
\qq\qq\qq\qq\qq $regular_{min} \la \false$;\ \breturn
\qq\qq\qq \belsif $\s{x_R}[j] \approx \s{x_R}[i]$ \bthen
\qq\qq\qq\qq $regular_{min} \la \false$;\ \breturn
$regular_{min} \la \true$;\ \breturn{}
}
\caption{Determine whether ($regular_{min} = \true$) or not ($regular_{min} = \false$) a reduced string $\s{x_R}[1..m]$ of scope $S(\s{x_R}) > 1$ on alphabet $\Sigma$ is regular.  If $\true$, on exit the string $\s{y_R}$ is the lex-least regular string of scope 1 on the integer alphabet $\Sigma' = \{1,2,\ldots,\sigma'\}$ that is isomorphic to \s{x_R}.}
\label{fig-regularmin}
\end{figure}

In the function $regular_{min}$, we first initialize each letter in $\s{y_R}[1..m]$ to $0$.  Then we scan $\s{x_R}$ from left to right, using $\s{y_R}$ to record previous matches. During this scan the following condition holds as long as \s{x_R} is regular:
\begin{equation}\label{e:condition_c}
    \s{x_R}[i] \approx \s{x_R}[j] \Leftrightarrow \s{y_R}[i] = \s{y_R}[j] \wedge \s{y_R}[i] \neq 0.
\end{equation}
If at position $i \in \{1,\dots,m\}$, we find $\s{y_R}[i]=0$ ---  that is, it was not part of a previous match --- we form the new character 
$\sigma' \la \sigma' \+ 1$. We then scan right in the strings $\s{x_R}[i+1..m]$ and $\s{y_R}[i+1..m]$ to see if condition (\ref{e:condition_c}) continues to hold. If it does not, we mark $\s{x_R}$ as indeterminate and exit; otherwise, whenever $\s{x_R}[j] \approx \s{x_R}[i]$ and $\s{y_R}[j] = 0$, we assign $\s{y_R}[j] \la \sigma'$.

In $regular_{min}$ the inner loops execute $\binom{m}{2}$ times, and checking if two letters match takes $\O(\sigma)$ time for $\s{x}$ in normal form, $\O(\sigma\log\sigma)$ otherwise. Therefore:

\begin{lemm}
\label{lem:regularminc}
Algorithm $regular_{min}$ runs in $\O(m^2\sigma)$ time, where $m = |\s{x_R}|$ and $1\leq m \leq 2^\sigma$, when \s{x_R} is in normal form; otherwise, $\O(m^2 \sigma\log\sigma)$ time, where $\O(\sigma\log\sigma)$ is the maximum time needed to sort the characters in two letters. \qed
\end{lemm}
\begin{lemm}
\label{lemm:yRy}  When $\s{x}[1..n]$ is regular, the (lex-least) corresponding \s{y} can be computed from \s{y_R} in $\Theta(n)$ time.
\end{lemm}
\begin{proof}
As noted above, the letters in $\s{x_R}[1..m]$ occur in order of their first appearance in \s{x}.  Moreover, the letters in \s{y_R} are assigned in increasing order $1,2,\ldots,\sigma'$ of occurrence of the positions of new letters in $\s{x_R}$.
Thus, to compute \s{y}, we first form $\xRpos[1..n]$, in which $\xRpos[j]$ is the position at which the letter $\s{x}[j]$ occurs in \s{x_R}. Therefore, the elements of $\xRpos$ are integers in the range $\{1 \dots m\}$. Clearly, $\xRpos[1..n]$ can be computed in $\Theta(n)$ time.
Then for every $i \in \{1 \dots n\}$, we compute
$$\s{y}[i] \la \s{y_R}\big[\xRpos[i]\big],$$
yielding a lex-least string of scope 1 isomorphic to \s{x}.
\end{proof}

\begin{lemm}
\label{lem:regularc} Algorithm REGULAR runs in $\Theta(n)$ time for a string $\s{x}[1..n]$ from a constant alphabet.
\end{lemm}
\begin{proof}
The running time of REGULAR is equal to the total time required for
(1) construction of \s{x_R} from \s{x};
(2) execution of $regular_{min}$;
(3) construction of \s{y} from \s{y_R}.

As we have seen (Claim 1), (1) requires $\Theta(n)$ time, while 
Lemma~\ref{lem:regularminc} tells us that (2) requires $\O(m^2 \sigma\log\sigma)$ time. However, since $m$ and $\sigma$ are constants, (2) requires constant time. Further, as explained in Lemma~\ref{lemm:yRy},
\s{y} can be computed from \s{y_R} in $\Theta(n)$ time.
Therefore, we conclude that the running time of REGULAR is $\Theta(n)$.\qed
\end{proof}

At a cost of adding $\Theta(m^2)$ bits of additional storage, together with $\O(m^2\sigma)$ preprocessing time, the regularity of the reduced string \s{x_R} can actually be determined in $\O(m^2)$ time in all cases, as we now explain.
Essentially, we replace \s{x_R} by a symmetric adjacency matrix $M^{\approx}_{\s{x_R}} = M^{\approx}_{\s{x_R}}[1..m,1..m]$ based on the matches in \s{x_R}: for $1 \le i,j \le m$,
\begin{eqnarray}
M^{\approx}_{\s{x_R}}[i,j] = 
\begin{cases}
1 & \mbox{ if } \s{x_R}[i] \approx \s{x_R}[j]; \\ 
0 & \mbox{ otherwise}.
\end{cases}
\end{eqnarray}
Since each match may require $\O(\sigma)$ time, this construction can be performed in $\O(m^2\sigma)$ time in the worst case.

Then, simply replacing references to \s{x_R} in function {\it regular\_{min}} by references to $M^{\approx}_{\s{x_R}}$, we obtain the $\O(m^2)$ function {\it regular\_{min}\_matrix} (Figure~\ref{fig-regularminM}).

\begin{figure}[htb]
{\leftskip=0.5in\obeylines\sfcode`;=3000
\bfunc $regular_{min\_matrix}(M^{\approx}_{\s{x_R}},m;\s{y_R},\sigma')$ 
$\s{y_R}[1..m] \la 0^m;\ \sigma' \la 0$ 
\bfor $i = 1$ \bto $m$ \bdo
\qq \bif $\s{y_R}[i] = 0$ \bthen 
\qq\qq $\sigma' \la \sigma'\+ 1;\ \s{y_R}[i] \la \sigma'$
\qq\qq \bfor $j \la i\+ 1$ \bto $m$ \bdo
\qq\qq\qq \bif $M^{\approx}_{\s{x_R}}[i,j]=1$ \bthen
\qq\qq\qq\qq \bif $\s{y_R}[j] = 0$ \bthen{} $\s{y_R}[j] \la \sigma'$
\qq\qq\qq\qq \belse $regular_{min\_matrix} \la \false;\ \breturn$
\qq \belse
\qq\qq \bfor $j \la i\+ 1$ \bto $m$ \bdo
\qq\qq\qq \bif $\s{y_R}[j] = \s{y_R}[i]$ \bthen
\qq\qq\qq\qq \bif $M^{\approx}_{\s{x_R}}[i,j] = 0$ \bthen
\qq\qq\qq\qq\qq $regular_{min\_matrix} \la \false$;\ \breturn
\qq\qq\qq \belsif $M^{\approx}_{\s{x_R}}[i,j]=1$ \bthen
\qq\qq\qq\qq $regular_{min\_matrix} \la \false$;\ \breturn
$regular_{min\_matrix} \la \true$;\ \breturn{}
}
\caption{Determine in $\O(m^2)$ time whether ($regular_{min\_matrix} = \true$) or not ($regular_{min\_matrix} = \false$) an $m \times m$ bit matrix $M^{\approx}_{\s{x_R}}$ is regular. In the former case, the algorithm computes, corresponding to $M^{\approx}_{\s{x_R}}$, a lex-least regular string \s{y_R} of scope 1 on the integer alphabet $\Sigma' = \{1,2,\ldots,\sigma'\}$.
}
\label{fig-regularminM}
\end{figure}

Effectively, $M^{\approx}_{\s{x_R}}$ defines an undirected graph $G^{\approx}_{\s{x_R}} = (V,E)$ on $m$ vertices, where there exists an edge $(i,j) \in E$ if and only if $M^{\approx}_{\s{x_R}}[i,j] = 1$.  The match relation $\approx$ is necessarily both \itbf{reflexive} ($M^{\approx}_{\s{w}}[i,i] = 1$) and \itbf{symmetric} ($M^{\approx}_{\s{w}}[i,j] = M^{\approx}_{\s{w}}[j,i])$ for any string \s{w}.  If in addition our string $\s{w} = \s{x_R}$ is regular, it follows immediately from our definition that the match relation is transitive:
$$M^{\approx}_{\s{x_R}}[i,j_1] = M^{\approx}_{\s{x_R}}[i,j_2] \Longleftrightarrow M^{\approx}_{\s{x_R}}[j_1,j_2] = M^{\approx}_{\s{x_R}}[i,j_1].$$
Thus for a regular string \s{x_R}, algorithms $regular_{min}$ and $regular_{min\_matrix}$ effectively compute the \itbf{transitive closure} of the match relation $\approx$ on \s{x_R} (equivalently, of the adjacent vertices in $G^{\approx}_{\s{x_R}}$); that is, they identify the maximum sets of positions that match in \s{x_R} (the cliques in $G^{\approx}_{\s{x_R}}$).

In the literature the Floyd-Warshall algorithms \cite{Cormen01, Floyd62} are the standard for computing the transitive closure on a graph $G = (V,E)$, with $|V| = m$.  These algorithms execute in time $\O(m^3)$, and it is a long-standing open problem whether a faster algorithm exists.  They are very general, not assuming the reflexive and symmetric conditions that hold in our case.  Nevertheless, we have:
\begin{lemm}\label{lem:tc}
For a relation on $m$ objects that is reflexive and symmetric, the transitive closure can be computed in $\O(m^2)$ time. \qed
\end{lemm}


\section{Strings from the Maximal Palindrome Array}
\label{sect-maxpal}
Given $\s{x} = \s{x}[1..n]$, a \itbf{substring} $\s{u} = \s{x}[i..j]$, $1 \le i \le j \le n$, of length $\ell = j\- i\+ 1$
is said to be a \itbf{palindrome}
if $\s{x}[i\+ h] \approx \s{x}[j\- h]$ for every $h \in 0..\floor{\ell/2}$;
a \itbf{maximal palindrome} if
one of the following holds:
$i = 1$, $j = n$, or
$\s{x}[i\- 1] \not\approx \s{x}[j\+ 1]$.
The \itbf{centre} of a palindrome \s{u} is at position $i \+ \frac{\ell\- 1}{2}$,
which is not an integer for odd $\ell$.
Thus to facilitate discussion of the palindromic structure of \s{x},
we form the string
$$\s{x^*}[1..m] = \# x_1\# x_2 \# \cdots \# x_n\#,$$
where $\# \not\in \Sigma$ and $m = 2n\+ 1$.
Now every palindromic substring of \s{x} of length $\ell \ge 1$
maps into a palindromic substring of \s{x^*} of odd length $d = 2\ell\+ 1$,
and every palindrome in \s{x^*} has an integral centre position.
We call $d$ the \itbf{diameter} and
$r = \floor{d/2}$ the \itbf{radius}
of the palindrome in \s{x^*}.

We can now define the \itbf{maximal palindrome array} $\MP = \MP_{\s{x^*}}$ of \s{x^*}:
for every $i \in \{1, \dots, m\}$,
if $\s{x^*}[i] = \#$ and $\s{x^*}[i\- 1] \not\approx \s{x^*}[i\+ 1]$, then $\MP[i] = 0$ (radius zero);
otherwise, $\MP[i]\geq 1$ is the radius of the maximal palindrome centred at position $i$.
We assume that $\s{x^*}[0]=\s{x^*}[m+1]=\emptyset$. For example, $\MP_{\s{x^*}}$ derived from $\s{x} = aabac$ is as follows:
\begin{equation}
\label{ex1}
\begin{array}{rcccccccccc}
\scriptstyle 1 & \scriptstyle 2 & \scriptstyle 3 & \scriptstyle 4 & \scriptstyle 5 & \scriptstyle 6 & \scriptstyle 7 & \scriptstyle 8 & \scriptstyle 9 & \scriptstyle 10 & \scriptstyle 11 \\ 
\s{x^*} = \# & a & \# & a & \# & b & \# & a & \# & c & \# \\ 
\MP_{\s{x^*}} = 0 & 1 & 2 & 1 & 0 & 3 & 0 & 1 & 0 & 1 & 0
\end{array}
\end{equation}
The most general form of the palindrome array is given by
\begin{equation}
\label{generalform}
\MP = 0 i_2 i_3 \cdots i_{m-1} 0,
\end{equation}
where for every $j \in \{2,\dots,m-1\}$:
\begin{itemize}
\item[(a)]
$i_j \in (1\- j\bmod 2)..\min(j\- 1,m\- j)$;
\item[(b)]
$i_j$ is odd if and only if $j$ is even.
\end{itemize}

Any array satisfying (\ref{generalform}) is said to be \itbf{feasible}.

\begin{lemm}
\label{lemm-exist}
There exists a string corresponding to
every feasible palindrome array.
\end{lemm}
\begin{proof}
Suppose that a feasible palindrome array  $\MP = \MP_{\s{x^*}}[1..m]$ is given
for some odd positive integer $m$.
We show how to construct a corresponding \s{x^*}.  
First, for every odd $c \in \{1,\dots,m\}$,
we may of course assign 
$\s{x^*}[c] \la \#$.
Suppose that the even positions $c$ in \s{x^*} are initially empty.
For every $c \in \{3, \dots, m\- 2\}$ such that 
$\MP[i] = r \ge 2$, add a unique character
to each pair of positions $c\- k,c\+ k$
in \s{x^*},
where
\begin{itemize}
    \item[$\bullet$] ($c$ even, $r$ odd) 
    $k = 2,4,\ldots,r\- 1$;
    \item[$\bullet$] ($c$ odd, $r$ even)
    $k =1,3,\ldots,r\- 1$.
\end{itemize}
Finally, assign a unique character to each position that remains empty. The resulting (perhaps indeterminate) string will have MP as its palindrome array. \qed
\end{proof}
Applying the construction of Lemma~\ref{lemm-exist} to
$\MP_{\s{x^*}} = 0 1 2 1 0 3 0 1 0 1 0$ given in (\ref{ex1}) yields a string such as:
\begin{equation}
\label{ex2}
\begin{array}{rcccccccccc}
\scriptstyle 1 & \scriptstyle 2 & \scriptstyle 3 & \scriptstyle 4 & \scriptstyle 5 & \scriptstyle 6 & \scriptstyle 7 & \scriptstyle 8 & \scriptstyle 9 & \scriptstyle 10 & \scriptstyle 11 \\ 
\s{x^*} = \# & a & \# & \{a,b\} & \# & c & \# & b & \# & d & \# \\ 
\MP_{\s{x^*}} = 0 & 1 & 2 & 1 & 0 & 3 & 0 & 1 & 0 & 1 & 0
\end{array}
\end{equation}
Note that, since the triple $(4,2,8)$ is
intransitive ($\s{x}[2] \not\approx \s{x}[8]$), \s{x^*} is indeterminate.
However, as we have seen in (\ref{ex1}), the regular string
$\# a \# a \# b \# a \# c \#$
has the same palindrome array:
a palindrome array can
correspond to both a regular and an
indeterminate string!

To each position $c \in \{1, \dots, m\}$ of a feasible $\MP$ array we associate a pair of integers $(i,j)$ such that $i=c-\MP[c] - 1$ and $j= c+ \MP[c] + 1$; then, provided $0<i,j< m+1$, we must have $\s{x^*}[i] \not \approx \s{x^*}[j]$. 
We call this pair $(i,j)$ the \itbf{forbidden pair} with respect to $c$, and the characters at $\s{x^*}[i], \s{x^*}[j]$ \itbf{forbidden characters} with respect to $c$. For processing purposes, it will be convenient to assume that $\s{x^*}[0]=\s{x^*}[m+1]=\emptyset$.  We denote by $\FP$ the set of all forbidden pairs with respect to each centre $c \in \{1, \dots, m\}$. See (\ref{ex:nonrmp})
for an example. 

Consider the feasible $\MP$ array given in~(\ref{ex:nonrmp}). 
This array does not correspond to any regular string because the triple $(6,2,4)$ is intransitive ($\s{x^*}[2] \not \approx \s{x^*}[4]$), since the pair $(2,4)$ is a forbidden pair with respect to the centre $3$. 
Indeed, any triple of $\s{x^*}$ which includes a forbidden pair is intransitive, resulting in an indeterminate string.
However, if a given $\MP$ array does correspond to a regular string, we will say that it is \itbf{regular}.

\begin{equation}
\label{ex:nonrmp}
\begin{array}{rcccccc}
c \qq 1 & 2 & 3 & 4 & 5 & 6 & 7 \\
MP \qq 0 & 1 & 0 & 3 & 2 & 1 & 0 \\
\s{x^*} \qq \# & 1 & \# & \{2,3\} & \# & \{1,3\} & \# \\
 \FP \qq (0,2) & (0,4) & (2,4) & (0,8) & (2,8) &  (4,8) & (6,8)
\end{array}
\end{equation}

In order to characterize MP arrays, it is useful to introduce
Manacher's condition \cite{manacher:75}, restated in \cite{IIBT10},
and further restated here.

In $\MP = \MP_{\s{x^*}}$, we consider each centre $c$
of a palindrome of radius $r = \MP[c]$, where for $c$ even,
$r \ge 3$, and for $c$ odd, $r \ge 2$.
Then we must have $\s{x^*}[c\- k] \approx \s{x^*}[c\+ k]$,
where $1 \le k \le r$.
We then have
\begin{defn}\label{defM} [Manacher's condition]
Let $r_{\ell} = \MP[c\- k]$ and $r_r = \MP[c\+ k]$, for each $1 \le k \le r$, where $r = \MP[c]$.
Every position $c$ in a regular palindrome array
$\MP = \MP_{\s{x^*}}$ satisfies the following:
\begin{itemize}
    \item[(a)]
    \bif $r_{\ell} \ne r\- k$ \bthen $r_r = \min(r_{\ell},r\- k)$
    \belse $r_r \ge r_{\ell}$;
    \item[(b)]
    \bif $r_r \ne r\- k$ \bthen $r_{\ell} = \min(r_r,r\- k)$
    \belse $r_{\ell} \ge r_r$.
\end{itemize}
\end{defn}

For completeness, we outline here Manacher's algorithm to construct an $\MP$ array for a given string $\s{x^*}$ of scope $1$. The algorithm scans $\s{x^*}$ from left to right and evaluates centres to compute $\MP[c]$. While evaluating a centre $c$ and each value of the range $k$, it checks whether $r_{\ell}=r-k$; if so, it stops evaluating centre $c$ and moves to evaluating centre $c\+ k$; otherwise it assigns $r_r \la \min(r_{\ell},r\- k)$ and continues evaluating the same centre $c$ and the next range value $k+1$.

In \cite{manacher:75} it is shown that, for every palindrome in string
\s{x^*} such that $S(\s{x^*}) = 1$,
Manacher's condition must hold.
Thus, by Lemma~\ref{lemmreg}, Manacher's condition holds for every
regular string; that is, for every string whose triples
are all transitive.

On the other hand, note that in example (\ref{ex:nonrmp}),
for $c = 4$ we have
$$r = \MP[4] = 3,\ k = 1,\ r_{\ell} = \MP[3] = 0,\ r_r = \MP[5] = 2,$$
so that, by Definition~\ref{defM}(a), since $r_{\ell} \ne r\- k$,
therefore $r_r = \min(0,2) = 0$ should hold, which is false.
Thus for this indeterminate string, Manacher's condition does {\it not} hold.

The above suggests extending the construction of Lemma~\ref{lemm-exist}
so as to yield, whenever possible, a regular string
corresponding to a given palindrome array.
See Figure~\ref{fig-construct}. 

The procedure $construct$ takes a feasible $\MP$ array as input and returns a lex-least regular string $\s{x^*}$ if $\MP$ is regular; otherwise it returns an indeterminate string. This procedure is based on the Manacher's algorithm presented in~\cite{manacher:75}. 
To ensure that the resulting string, if regular, is a lex-least string, it maintains for each centre $c$ a set $\FS[c]$ containing the indices of the forbidden characters. 

\begin{figure}[h!]
	{\leftskip=0.5in\obeylines\sfcode`;=3000
		\bfunc $construct(\MP,m;\ \s{x^*},\sigma,regular)$
		\com{$\MP[1..m]$ is assumed to be feasible, and $m=|\s{x^*}|$}.
		\com{$\FS[1..m]$ gives the index position of all forbidden characters w.r.t each centre of the $\MP$ array.}
		\com{$empty[1..m]$ is a Boolean array initially set to true, where }
		\com{$empty[c]=true$ whenever $update\FS$ has not been called for centre $c$.}
		$regular \la \true;\ \sigma \la 0;\ \s{x^*}[1] \la \#;\ \FS[1] \la \emptyset$
		\bfor $c \la 2$ \bto $m\m 1$ \bstep 2 \bdo \Comment{Initialize $\s{x^*}$ and $\FS$ array.}
		\qq $\s{x^*}[c] \la 0;\ \FS[c] \la \emptyset$ 
		\qq $\s{x^*}[c+1] \la \#;\ \FS[c+1] \la \emptyset$
		$empty[1..2]=false ;\ empty[3..m] = true$
		$\sigma \la \sigma\+ 1;\ \s{x^*}[2] \la \sigma$
		\bfor $c \la 3 $ \bto $m\m 1$ 
		\qq $r \la \MP[c]$
		\qq \bif $regular$ \band $empty[c]$ \bthen 
		\qq\qq $empty[c] \la false;\ update\FS(c,r, m,\FS)$
		\qq \bif $r > 1$ \bthen
	    \qq\qq $nextc \la c\+ r \+ 1$
		\qq\qq \bfor $k \la r$ \bdownto $1$ \bdo

        \qq\qq\qq $r_{\ell}=\MP[c\m k] ;\ r_{r}=\MP[c\+k]$
        \qq\qq\qq \com{$MC$ is a function that tests Manacher's condition}
		\qq\qq\qq \com{for centre $c$, radius $r$ and range $k$.}
		\qq\qq\qq \bif $regular$ and $MC(c,r,k)$ \bthen
		\qq\qq\qq\qq \bif $empty[c\+ k]$ \bthen 
		\qq\qq\qq\qq\qq $empty[c\+ k] \la false ;\ update\FS(c \+ k,r_r, m,\FS)$
		\qq\qq\qq\qq \bif $\s{x^*}[c\+ k] = 0$ \bthen $\s{x^*}[c\+ k] \la \s{x^*}[c \m k]$ 
        \qq\qq\qq\qq \bif $r_{\ell} = r-k$ and $r_r > r_{\ell}$ \bthen  $nextc\la c+k$
		\qq\qq\qq \belse
		\qq\qq\qq\qq $regular \la false $
		
		\qq\qq\qq\qq \bif $c \+ k \mod 2 = 0$ \band $\s{x^*}[c\m k] \not\approx \s{x^*}[c\+ k]$ \bthen  
		\qq\qq\qq\qq\qq $\sigma \la \sigma\+ 1;\ \s{x^*}[c\m k] \stackrel{+}{\la} \sigma;\ \s{x^*}[c\+ k] \stackrel{+}{\la} \sigma$
		\qq\qq\qq $k \la k \m 1$
		\qq \bif $\s{x^*}[c] = 0$ \bthen \Comment{Fill empty position at centre.}
		\qq\qq \bif $regular$ \bthen 
		\qq\qq\qq $\s{x^*}[c] \la $ Lex-least character not in $\FS[c]$
		\qq\qq \belse $\sigma \la \sigma\+ 1;\ \s{x^*}[c] \stackrel{+}{\la} \sigma$
		\qq \bif $regular$ \bthen $c \la nextc$
		\qq \belse $c \la c \+ 1$
    \breturn $\s{x^*}$
	}

\caption{Given a feasible palindrome array
MP, $construct$ if possible produces a corresponding lex-least regular string \s{x^*} of scope 1
on an integer alphabet $\Sigma = \{1,2,.\ldots\}$.
If such a string is not possible, then on exit $regular = \false$ and \s{x^*}
is an indeterminate string on $\Sigma$ corresponding to MP.
}
\label{fig-construct}
\end{figure}

\begin{figure}[h!]
	{\leftskip=0.5in\obeylines\sfcode`;=3000
		\bproc $update\FS(c, r, m, \FS)$
		\bif $ c\m r \m 1 \neq 0$ \band $c \+ r \+ 1  \neq m \+ 1$ \bthen
		\qq $\FS[c\+ r \+ 1] \stackrel{+}{\la} c-r-1$
	}
\caption{Given centre $c$ and its radius $r$, $update\FS$ updates the $\FS$ array at the left position in the forbidden pair with respect to $c$.}
\label{fig-updateFS}
\end{figure}

In procedure $construct$ the output string is first initialized with $0$ at even positions and $\#$ at odd positions. Then $\s{x^*}[2]$ is set to $1$, as we are interested in the lex-least regular $\s{x^*}$, if it exists. Next we examine a centre $3 \leq c \leq m \m 1$ and update the $\FS[c]$ set based on its forbidden pair. If the string $\s{x^*}[1..c\m 1]$ constructed so far is regular and the centre $c$ and range $k$ satisfy Manacher's condition, we continue to construct the regular string by copying the previously filled letter $\s{x^*}[c\m k]$ to its corresponding matching position $c\+ k$. Whenever there is a choice of filling an empty position --- that is, when $\s{x^*}[c]=0$ --- the lex-least character which is not in the forbidden set of characters $\FS[c]$ is chosen. Note that we use the same strategy as presented in Manacher's algorithm~\cite{manacher:75} to avoid evaluating every centre if $\MP$ is regular. However, unlike Manacher's algorithm, instead of evaluating a new centre when $r_{\ell}=r-k$ and $r_r > r_{\ell}$, we record the next centre to be evaluated, and finish evaluating the current centre. After which we move to the recorded centre and check if Manacher's condition is satisfied.

Furthermore, to avoid invoking the $update\FS$ procedure repeatedly on a centre $c$, we use a Boolean array $empty$ of length $m$, such that $empty[c] = true$ if and only if the centre $c$ has not been used to update the $\FS$ array; otherwise, we set $empty[c] \la false$. If a given centre $c$ and range $k$ do not satisfy Manacher's condition, $regular$ is set to $false$ and every subsequent letter match including the current one is achieved by adding a new character to the letters at positions $c\m k$ and $c\+ k$, if they are even. Furthermore, for the remainder of the $\MP$ array all centres are examined. 

\begin{lemm}\label{lemm-regconstx}
Let $\s{x^*}$ be the string produced by procedure $construct$. Then
$S(\s{x^*})=1 \Leftrightarrow \s{x^*}$ is regular.
\end{lemm}
\begin{proof}
($\Rightarrow$) Since $S(\s{x^*})=1$,
the regularity of $\s{x^*}$ is immediate, as
for all positions $1 \leq i,j \leq m$, $\s{x^*}[i] \approx \s{x^*}[j] \Leftrightarrow \s{x^*}[i] = \s{x^*}[j]$. Then any triple in $\s{x^*}$ is transitive by the associative property of equality. Therefore $\s{x^*}$ is regular.

($\Leftarrow$) 
$\s{x^*}$ is regular only if Manacher's condition holds for every centre $c$, radius $r$ and range $k$. In such a case $construct$ either copies existing letters of scope one or fills an empty position with a new character to produce $\s{x^*}$. Therefore $S(\s{x^*})=1$.
\qed
\end{proof}

\begin{thrm}\label{thrm-const}
Let $\s{x^*}$ be the string produced by the procedure $construct$ on an input $\MP$. Then
$\s{x^*}$ is regular $\Leftrightarrow$ $\MP$ is regular.
\end{thrm}
\begin{proof}
($\Rightarrow$) $\s{x^*}$ is regular only if Manacher's condition holds for every centre $c \in \{1, \ldots, m\}$, and for every corresponding radius $r$ and range $k$. Since Manacher's condition holds, it follows from the result in~\cite{manacher:75} that $\MP$ is regular. 

($\Leftarrow$) From~\cite{manacher:75}, MP is regular only if Manacher's condition holds. Since the condition holds, during the execution of $construct$ an empty position is filled either by copying previously occurring letters of scope one or by a new character. Thus the resulting string is of scope $1$, so that from Lemma~\ref{lemm-regconstx}, $\s{x^*}$ is regular.
\qed
\end{proof}

Using the two examples given below in (\ref{ex3}) and (\ref{ex4}), we illustrate the execution of procedure $construct$. 

In (\ref{ex3}) the positions corresponding to all odd centres are filled with $\#$ and positions corresponding to all even centres are filled with $0$. Furthermore, $\FS[1..m] \la \emptyset$, $empty[1..2] \la false$ and $empty[3..m] \la true$. Then
for centre $c = 2$, $\s{x^*}[2] \la 1$, and for centre $c = 3$, $FS[4] \la \{2\}$ and $empty[3] \la false$. At centre $c=4$, $empty[4] \la false$, and since $r=3$ we check if Manacher's condition holds for $k=3,2,1$. As we do this, we also check whether $empty[c+k]$ is $true$; if so, we set it to $false$ and compute the $\FS$ set w.r.t to these centres, thereby setting $\FS[6]$ and $\FS[8]$ to $\{4\}, \{4,6\}$ respectively. Since Mancher's condition holds, the setting at position $\s{x^*}[2]$ is copied into $\s{x^*}[6]$; next, since position $\s{x^*}[4]=0$, it is filled with the next lex-least character which is not equal to a character at any index position in the forbidden set $\FS[4]=\{2\}$ --- that is, $2$.

Similarly for centre $c = 8$, $empty[8..15] \la false$, $\FS[10] \la \{8\}, \FS[12] \la \{8,10\}$. Since Manacher's condition holds for all $k=7, 6, \ldots, 1$, $\s{x^*}[2]$, $\s{x^*}[4]$ and $\s{x^*}[6]$ are copied into $\s{x^*}[14],\s{x^*}[12]$ and $\s{x^*}[10]$, respectively. Finally, since $\s{x^*}[8]=0$, the position is filled with the next lex-least character $3$ which is not equal to a character at any index position in the forbidden set $\FS[8]=\{4,6\}$. The resulting string $\s{x^*}$ is regular and lex-least.
\begin{small}
\begin{equation}
\label{ex3}
\begin{array}{rcccccccccccccccc}
\scriptstyle 1 & \scriptstyle 2 & \scriptstyle 3 & \scriptstyle 4 & \scriptstyle 5 & \scriptstyle 6 & \scriptstyle 7 & \scriptstyle 8 & \scriptstyle 9 & \scriptstyle 10 & \scriptstyle 11 & \scriptstyle 12 & \scriptstyle 13 & \scriptstyle 14 & \scriptstyle 15 \\
MP = 0 & 1 & 0 & 3 & 0 & 1 & 0 & 7 & 0 & 1 & 0 & 3 & 0 & 1 & 0 \\
\s{x^*} = \# & 1 & \# & 2 & \# & 1 & \# & 3 & \# & 1 & \# & 2 & \# & 1 & \# \\
FS = \emptyset & \emptyset & \emptyset & \{2\} & \emptyset & \{4\} & \emptyset & \{4, 6\} & \emptyset & \{8\} & \emptyset & \{8,10\} & \emptyset & \{12\} & \emptyset \\
\end{array}
\end{equation}
\end{small}

Observe that the $\MP$ array in (\ref{ex4}) only differs from the $\MP$ array in (\ref{ex3}) at position $12$. However, the $\MP$ array in (\ref{ex4}) does not correspond to any regular string, since, as we shall see, it fails Manacher's condition for $c=8,\  r=7$, and $k=4$. In (\ref{ex4}) $construct$ produces the same string as in (\ref{ex3}) up until $c=7$. Then at centre $c = 8$, $\s{x^*}[2]$ is first copied into $\s{x^*}[14]$. However, since for $k=4$ the $\MP$ array fails Manacher's condition, a new lex-least character $3$ is added to the letters at positions $4$ and $12$. Then at positions $6$ and $10$, a new character $4$ is added to both $\s{x^*}[6]$ and $\s{x^*}[10]$. Finally, since $\s{x^*}[8]=0$, it is filled with another new character $5$. 
{\small
\begin{equation}
\label{ex4}
\begin{array}{rcccccccccccccccc}
\scriptstyle 1 & \scriptstyle 2 & \scriptstyle 3 & \scriptstyle 4 & \scriptstyle 5 & \scriptstyle 6 & \scriptstyle 7 & \scriptstyle 8 & \scriptstyle 9 & \scriptstyle 10 & \scriptstyle 11 & \scriptstyle 12 & \scriptstyle 13 & \scriptstyle 14 & \scriptstyle 15 \\
MP = 0 & 1 & 0 & 3 & 0 & 1 & 0 & 7 & 0 & 1 & 0 & 1 & 0 & 1 & 0 \\
\s{x^*} = \# & 1 & \# & \{2,3\} & \# & \{1,4\} & \# & 5 & \# & 4 & \# & 3 & \# & 1 & \#\\
\end{array}
\end{equation}
}

We now discuss the running time complexity of $construct$. This  is dependent on the input $\MP$ array. If $\MP$ is regular, then $construct$ yields a lex-least regular string in $\mathcal{O}(n)$ time; otherwise it yields an indeterminate string in $\mathcal{O}(n^2)$ time. We remark that to avoid having $\sigma$ in the running time, we use the strategy discussed in~\cite{IIBT10} to fill the lex-least character not in the list of forbidden characters at centre $c$. Here we briefly outline the strategy: The paper \cite{IIBT10} uses a bit vector $F$ of length $n$ where the element $F[h]$ corresponds to the $h$-th lex-least least character. Initially all elements of $F$ are set to $0$. If a character $\s{x}[j]$ is forbidden w.r.t the character $\s{x}[i]$, and if $\s{x}[j]$ is the $h$-th lex-least character, then $F[h]$ is set to $1$. After finding all forbidden characters for $\s{x}[i]$ and setting the $F$ vector appropriately, the lex-least character for $\s{x}[i]$ is set equal to the $k$-th lex-least character, where $k$ is the first (or smallest) index such that $F[k]=0$.

If $\MP$ is regular, then similar to Manacher's algorithm, not every centre is evaluated to yield a regular string. Suppose $c$, $1 \leq c \leq m$, is the least centre for which the Manacher's test is performed and is satisfied. Then, two cases arise.  First, there exists a $k$ (in $construct$ we consider the minimum $k$), $1\leq k \leq r$, where the following subcondition of Manacher's condition (Definition~\ref{defM}) is satisfied: $r_{\ell} = r\m k$ and $r_r > r_{\ell}$.  Second, when the subcondition is not satisfied, all the palindromes that arise at centres $c'\leq c\+ r$ are contained in the palindrome at $c$ and also satisfy Manacher's condition. Therefore, to avoid duplication the next centre to be evaluated is $c\+r\+1$.

In the former case, however, the palindrome at centre $c+k$ extends beyond the right edge of the palindrome centred at $c$. Therefore, $c\+ k$ is the next centre to be evaluated. However, since $c$ satisfies Manacher's, a shorter palindrome of length $r\m k$ at centre $c\+k$ also occurs at centre $c\m k$. Therefore, we only need to check the positions occurring after the first $r\m k$ positions on both sides of $c\+k$. 
In both these cases a position in the $\MP$ array is visited at most twice.  Therefore, the running time of $construct$ to yield a regular string is $\mathcal{O}(n)$.

If $\MP$ is not regular, after the first centre that fails Manacher's test, every centre is evaluated, and as a result $construct$ yields an indeterminate string in $\mathcal{O}(n^2)$ time. However, note that the running time for an $\MP$ array that is not regular largely depends on the centre that fails the Manacher's test. If the centre lies close to the end of the array, then the indeterminate string is constructed in $\mathcal{O}(n)$ time; otherwise it takes $\mathcal{O}(n^2)$ time. Thus we have:

\begin{thrm}\label{thrm-rtconstruct}
Given a regular MP array of length $m = 2n\+ 1$,
procedure $construct$ yields a lex-least regular string in $\mathcal{O}(n)$ time.
\end{thrm}  

\begin{thrm}\label{thrm-rtconstruct}
Given a MP array of length $m = 2n\+ 1$ that is not regular,
procedure $construct$ yields an indeterminate string in $\mathcal{O}(n^2)$ time.
\end{thrm}  

We remark that our algorithm $construct$ is similar to the approach presented in~\cite[Algorithm~2]{IIBT10} to construct a regular string if $\MP$ is regular. The key differences between $construct$ and their approach is that their algorithm sometimes verifies an $\MP$ array which is not regular to be regular, and therefore outputs an incorrect lex-least regular string $\s{w}$. To correct this, they execute Manacher's algorithm~\cite[Algorithm~1]{IIBT10} on $\s{w}$ to compute its palindrome array $\MP_{\s{x}}$. Then they output $\s{w}$ if and only if $\MP_{\s{x}} = \MP$. We on the other hand use a direct approach --- our algorithm outputs a regular string if and only if the input $\MP$ array is regular. Furthermore, if $\MP$ is not regular, their algorithm simply exits, while $construct$ on the other hand yields an indeterminate string.

Consider the example shown in Figure~\ref{ex5}. For the given $\MP$ array, $construct$ produces the string $\s{x^*}$. Observe that not all the centres $4-11$ satisfy Manacher's condition. Furthermore, if a centre $c$ and a range $k$ fail Manacher's condition, it is not necessarily the case that the centre $c+k$ also fails it. For example, for centre $c=8$, and range $k=3$, Manacher's condition fails; however, the centre $c+k=11$ satisfies it. 

\begin{figure}
{\footnotesize
\begin{equation*}
\arraycolsep=2.2pt
\begin{array}{rcccccccccccccccc}
\scriptstyle 1 & \scriptstyle 2 & \scriptstyle 3 & \scriptstyle 4 & \scriptstyle 5 & \scriptstyle 6 & \scriptstyle 7 & \scriptstyle 8 & \scriptstyle 9 & \scriptstyle 10 & \scriptstyle 11 & \scriptstyle 12 & \scriptstyle 13 \\
MP = 0 & 1 & 0 & 3 & 2 & 3 & 6 & 3 & 4 & 1 & 2 & 1 & 0 \\
\s{x^*} = \# & \{1,5\} & \# & \{2,3, 4,6\} & \# & \{1,3, 7, 8 ,9\} & \# & \{4,7,10\} & \# & \{6,8,10,11\} & \# & \{5,9,11\} & \# \\
\FS = \emptyset & \emptyset & \emptyset & \{2\} & \emptyset & \emptyset & \emptyset & \{2\} & \emptyset & \{2\} & \emptyset & \{4,8\} & \emptyset
\end{array}
\end{equation*}
}
\caption{Given an $\MP$ array that is not regular, procedure $construct$ produces the indeterminate string $\s{x^*}$.}\label{ex5}
\end{figure}

As can be observed in the example shown in Figure~\ref{ex5}, the number of centres that fail Manacher's condition is of order $n$, where $n$ is the length of the string. Therefore, we are forced to evaluate $\mathcal{O}(n)$ centres in this example. Hence we state the following:

\begin{conj}\label{lem:con}
To construct an indeterminate string over an MP array of length $m = 2n\+ 1$ that is not regular
requires at least $\Omega(n^2)$ time.

\end{conj}

\section{Conclusion and Open Problems}
\label{sect-probs}

In this paper we have characterized regular strings as those whose triples are all transitive; the strings for which this is not true are called indeterminate.  A new approach based on the transitivity of matching among the regular letters of a string enables us to precisely define strings and eliminate false positive cases seen in the traditional definition of indeterminate strings. Based on this new definition, we present a linear time algorithm REGULAR that checks the regularity of the given string $\s{x}$ by checking the regularity of its corresponding reduced string $\s{x_R}$, and as a byproduct yields a lex-least regular string that is isomorphic to $\s{x}$ whenever $\s{x}$ is regular.

To check the regularity of $\s{x_R}$, we present two algorithms $regular_{min}$ and $regular_{min\_matrix}$, which  illustrate an interesting connection between the matching relation defined on the letters of $\s{x}$ and the computation of the relation's transitive closure. Many algorithms exist to compute the transitive closure of a relation in $\O(n^3)$ time. However, it is a long standing open problem whether this complexity can be reduced. We show that for important special cases, the transitive closure can be computed in $\mathcal{O}(n^2)$ time using the $regular_{min\_matrix}$ function.

Furthermore, we consider the reverse engineering problem of constructing a string (regular or indeterminate) from the given palindrome array. We present the algorithm $construct$ to output a lex-least regular string, if it exists, in $\mathcal{O}(n)$ time; otherwise the algorithm returns an indeterminate string in $\mathcal{O}(n^2)$ time. An immediate question arises whether an algorithm faster than $construct$ exists to output an indeterminate string. It would also be interesting to develop an algorithm that yields an indeterminate string over a minimum alphabet, corresponding to an $\MP$ array that is not regular. Furthermore, although the running time complexity of REGULAR is $\mathcal{O}(n)$, the constant hidden in it could in extreme cases be rather large -- is it possible to reduce the size of this constant?

More generally, a new definition of regularity in strings may affect
the relationship between strings and the data structures mentioned
in the Introduction that facilitate their use:
border array, cover array, prefix array, and suffix array. 
A string containing indeterminate letters may now actually be regular,
provided all its triples are transitive:
can these important arrays be computed accordingly in linear time?

Conversely, what effect does this new definition have on the
``reverse engineering'' problem related to these arrays?
We have presented in Section~\ref{sect-maxpal} an algorithm that
reverse engineers the palindrome array to a regular string,
whenever possible.
There exist algorithms to reverse engineer the other array noted
above --- can they be modified/extended to yield equivalent results?
Can the time requirements of these algorithms be reduced to near-linear,
except in pathological cases?

\section*{Acknowledgments}
We thank the reviewers for carefully reading the manuscript and providing insightful comments which led to substantial improvements in the manuscript.

\paragraph{Funding:}
Louza was supported by Grant $\#$2017/09105-0 from the S\~ao Paulo Research Foundation (FAPESP).
Mhaskar and Smyth were supported by Grant 36797 from the Natural Sciences
\& Engineering Research Council of Canada (NSERC).


\begin{thebibliography}{10}
\expandafter\ifx\csname url\endcsname\relax
  \def\url#1{\texttt{#1}}\fi
\expandafter\ifx\csname urlprefix\endcsname\relax\def\urlprefix{URL }\fi
\expandafter\ifx\csname href\endcsname\relax
  \def\href#1#2{#2} \def\path#1{#1}\fi

\bibitem{FP74}
M.~Fischer, M.~Paterson, String matching and other products, in: R.~Karp (Ed.),
  Complexity of Computation,, American Mathematical Society, 1974, pp.
  113--125.

\bibitem{A87}
K.~Abrahamson, Generalized string matching, SIAM Journal of Computing 16~(6)
  (1987) 1039--1051.

\bibitem{BS08}
F.~Blanchet-Sadri, Algorithmic Combinatorics on Partial Words, Chapman \& Hall
  CRC, 2008.

\bibitem{SW09a}
W.~F. Smyth, S.~Wang, A new approach to the periodicity lemma on strings with
  holes, Theoret. Comput. Sci. 410~(43) (2009) 4295 -- 4302.

\bibitem{IMR08}
C.~S. Iliopoulos, L.~Mouchard, M.~S. Rahman, A new approach to pattern matching
  in degenerate {DNA/RNA} sequences and distributed pattern matching, Math. in
  Computer Science (2008) 557--569.

\bibitem{IRVV10}
C.~Iliopoulos, M.~S. Rahman, M.~Vor\'{a}\v{c}ek, L.~Vagner, Finite automata
  based algorithms on subsequences and supersequences of degenerare strings, J.
  Discrete Algorithms 6~(2) (2010) 117--130.

\bibitem{IKP17}
C.~Iliopoulos, R.~Kundu, S.~P. Pissis, Efficient pattern-matching in
  elastic-degenerate strings, Proc. Language \& Automata Theory \& Applications
  (LATA 2017), LNCS 10168 (2017) 131--142.

\bibitem{HS03}
J.~Holub, W.~F. Smyth, Algorithms on indeterminate strings, Proc. 14th
  Australasian Workshop on Combinatorial Algs. (AWOCA) (2003) 36--45.

\bibitem{HSW05}
J.~Holub, W.~F. Smyth, S.~Wang, Fast pattern--matching on indeterminate
  strings, Proc. 16th Australasian Workshop on Combinatorial Algorithms (2005)
  415--428.

\bibitem{HSW06}
J.~Holub, W.~F. Smyth, S.~Wang, Hybrid pattern-matching algorithms on
  indeterminate strings, in: J.~W. Daykin, M.~Mohamed, K.~Steinhofel (Eds.),
  London Algorithmics and Stringology, King's College Texts in Algorithmics,
  2006, pp. 115--133.

\bibitem{HSW08}
J.~Holub, W.~F. Smyth, S.~Wang, Fast pattern--matching on indeterminate
  strings, J. Discrete Algorithms 6~(1) (2008) 37--50.

\bibitem{SW09}
W.~F. Smyth, S.~Wang, An adaptive hybrid pattern-matching algorithm on
  indeterminate strings, Internat. J. Foundations of Computer Science 20~(6)
  (2009) 985--1004.

\bibitem{NRR12}
S.~Nazeen, S.~M. Rahman, R.~Reaz, Indeterminate string inference algorithms, J.
  Discrete Algorithms 10 (2012) 23--34.

\bibitem{BRS09}
M.~F. Bari, M.~S. Rahman, R.~Shahriyar, Finding all covers of an indeterminate
  string in $o(n)$ time on average, in: J.~Holub, J.~{\v{Z}}{\v{d}}{\'a}rek
  (Eds.), Proceedings of the Prague Stringology Conference 2009, 2009, pp.
  263--271.

\bibitem{ARS15}
A.~Alatabbi, M.~S. Rahman, W.~F. Smyth, Inferring an indeterminate string from
  a prefix graph, J. Discrete Algorithms 32 (2015) 6--13.

\bibitem{CRSW15}
M.~Christodoulakis, P.~J. Ryan, W.~F. Smyth, S.~Wang, Indeterminate strings,
  prefix arrays and undirected graphs, Theoretical Comput.\ Sci. 600~(4) (2015)
  34--48.

\bibitem{BIST03}
H.~Bannai, S.~Inenaga, A.~Shinohara, M.~Takeda, Inferring strings from graphs
  and arrays., in: B.~Rovan, P.~Vojt's (Eds.), Symp. on Mathematical
  Foundations of Computer Science, Vol. 2747 of Lecture Notes in Computer
  Science, Springer, 2003, pp. 208--217.

\bibitem{HRSS18}
J.~Helling, P.~J. Ryan, W.~F. Smyth, M.~Soltys, Constructing an indeterminate
  string from its associated graph, Theoret. Comput. Sci. 710 (2018) 88--96.

\bibitem{FLRSSY99}
F.~Franek, W.~Lu, P.~J. Ryan, W.~F. Smyth, Y.~Sun, L.~Yang, Verifying a border
  array in linear time (preliminary version), Proc. 10th Australasian Workshop
  on Combinatorial Algs. (AWOCA) School of Computing, Curtin University of
  Technology (1999) 26--33.

\bibitem{FGLRSSY02}
F.~Franek, S.~Gao, W.~Lu, P.~J. Ryan, W.~F. Smyth, Y.~Sun, L.~Yang, Verifying a
  border array in linear time, J. Combinatorial Maths. and Combinatorial
  Comput. 42 (2002) 223--236.

\bibitem{CCR09}
J.~Cl\'ement, M.~Crochemore, G.~Rindone, Reverse engineering prefix tables, in:
  Proc. 26th Symp. on Theoretical Aspects of Computer Science, 2009, pp.
  289--300.

\bibitem{DFHSS18}
J.~W. Daykin, F.~Franek, J.~Holub, A.~S.~M.~S. Islam, W.~F. Smyth,
  Reconstructing a string from its {L}yndon arrays, Theoret. Comput. Sci. 710
  (2018) 44--51.

\bibitem{IIBT10}
T.~I, S.~Inenaga, H.~Bannai, M.~Takeda, Counting and verifying maximal
  palindromes, Proc. 17th Symposium on String Proocessing \& Information
  Retrieval (SPIRE), LNCS 6393 (2010) 135--146.

\bibitem{prochazka19}
P.~Proch{\'{a}}zka, J.~Holub, On-line searching in {IUPAC} nucleotide
  sequences, in: E.~D. Maria, A.~L.~N. Fred, H.~Gamboa (Eds.), Proceedings of
  the 12th International Joint Conference on Biomedical Engineering Systems and
  Technologies {(BIOSTEC} 2019) - Volume 3: BIOINFORMATICS, Prague, Czech
  Republic, February 22-24, 2019, SciTePress, 2019, pp. 66--77.

\bibitem{btt2013}
L.~Huang, V.~Popic, S.~Batzoglou, {Short read alignment with populations of
  genomes}, Bioinformatics 29~(13) (2013) i361--i370.

\bibitem{gray53}
F.~Gray, Pulse code communication.

\bibitem{Cormen01}
T.~H. Cormen, C.~E. Leiserson, R.~L. Rivest, C.~Stein, Introduction to
  Algorithms, 2nd Edition, MIT Press, Cambridge, MA, 2001.

\bibitem{Floyd62}
R.~W. Floyd, Algorithm 97: Shortest path, Communications of the ACM 5~(6)
  (1962) 345.

\bibitem{manacher:75}
G.~Manacher, {A} new linear-time on-line algorithm for finding the smallest
  initial palindrome of a string, J. Assoc. Comput. Mach. 22 (1975) 346--351.

\end{thebibliography}

\end{document}